\documentclass[conference,letterpaper]{IEEEtran}

\usepackage{balance}

\usepackage[utf8]{inputenc} 
\usepackage[T1]{fontenc}
\usepackage{url}
\usepackage{ifthen}
\usepackage{graphicx}

\usepackage{graphicx} %
\usepackage{amsthm}
\usepackage{mathrsfs}
\theoremstyle{remark}
\newtheorem{definition}{Definition}
\newtheorem{theorem}{Theorem}
\newtheorem{corollary}{Corollary}

\newtheorem{remark}{Remark}
\newtheorem{example}{Example}
\usepackage{amssymb}
\usepackage{verbatim}

\usepackage{xcolor}

\usepackage[square,numbers,sort&compress]{natbib}
\usepackage{bbm}

\newcommand{\identity}{\mathbbm{1}}

\usepackage{physics}

\usepackage{bm}
\usepackage{caption}
\usepackage{braket}

\usepackage{enumerate}

\begin{document}
\title%
{Quantum Secret Sharing Rates} 

\author{%
  \IEEEauthorblockN{Gabrielle Lalou}
  \IEEEauthorblockA{Department of Electrical Engineering \\
                    Télécom Paris, Institut Polytechnique de Paris\\
                    Palaiseau, France\\
                    Email: gabrielle.lalou@telecom-paris.fr}
  \and
  \IEEEauthorblockN{Husein Natur and Uzi Pereg}
  \IEEEauthorblockA{Helen Diller Quantum Center\\ 
                    Technion - Israel Institute of Technology\\
                    Haifa 3200003, Israel\\
                    Email: husein.natur@campus.technion.ac.il,\\ uzipereg@technion.ac.il}
}

\maketitle

\begin{abstract}
   This paper studies the capacity limits for quantum secret sharing (QSS). %
The goal of a QSS scheme is to distribute a quantum secret among multiple participants, such that only authorized parties can recover it through collaboration, while no information can be obtained without such collaboration.
   We introduce an information-theoretic model for the rate analysis of QSS
and its relation to compound quantum channels,
   following a similar approach as of
Zou \emph{et al.} (2015) on classical secret sharing. We establish a regularized characterization for the QSS capacity, and determine the capacity for QSS with dephasing noise.
\end{abstract}

\section{Introduction}

In modern communication networks, 
confidentiality %
is a fundamental requirement \cite{bloch2011physical}.
Quantum communication protocols address this need by leveraging physical laws to provide information-theoretic security.
Among these, QKD %
is the most technologically mature, though many other security primitives rely on 
quantum principles
\cite{zapatero2025implementation,graifer2023quantum,farre2025entanglement}.

The wiretap channel model, introduced by Wyner \cite{wyner1975wire}, and the secrecy capacity of its quantum analogue, have been studied extensively \cite{cai2004quantum, devetak2005private, tikku2020non}. The secrecy capacity quantifies %
how much information can be securely transmitted in the presence of an eavesdropper. %
Further %
network extensions have been studied as well, including semantic security \cite{boche2022semantic}, layered secrecy \cite{pereg2021key}, covert communication \cite{tahmasbi2020covert}, tactile communication \cite{hassanpour2025quantum}, and unreliable entanglement assistance \cite{lederman2024secure}.

Suppose a company must secure the password to its safe. 
To ensure  security, the company may distribute partial information  among $K$ colleagues, each receiving a share of the secret. %
The secret sharing model was originally introduced by Shamir~\cite{shamir1979share}  as a method for distributing a secret securely. %
In the $(t,K)$ threshold scheme, a secret is divided among $K$ participants in such a way that only a %
subset of $t$ or more participants can collaboratively reconstruct the secret. %

Quantum Secret Sharing (QSS) extends the principles of classical secret sharing to the quantum domain. It finds applications in several contexts, such as the creation of joint bank accounts containing quantum money \cite{wiesner1983conjugate, aaronson2009quantum},
and the secure execution of distributed quantum computations \cite{smith2000quantum}. 
The goal of QSS is to encode and distribute an arbitrary, unknown quantum state among several participants. Each participant holds a share. A \emph{qualified set} of participants is a subset capable of perfectly reconstructing the original secret through collaboration, whereas \emph{non-qualified} sets of participants 
gain no information about it. 
The no-cloning theorem, which asserts that an unknown quantum state cannot be perfectly copied, ensures that once a qualified set reconstructs the secret, it remains inaccessible to all remaining participants.
Furthermore, a $(t,K)$ threshold QSS (t-QSS)  scheme is only feasible for $2t>K$.
The collection of  qualified subsets is referred to as %
the \emph{access structure}. %
Several studies have examined QSS across different tasks \cite{%
gottesman2000theory,%
zhang2005multiparty,markham2008graph}. %
Among the earliest approaches,
Hillery et al. \cite{hillery1999quantum} (HBB99)
demonstrated that multipartite entanglement, specifically GHZ states, can serve as a viable quantum resource for sharing a classical secret.
Additional QSS and semi-QSS settings have also been explored, %
with classical messages \cite{%
guo2003quantum, xiao2004efficient, li2010semiquantum, Zhang_2011} and %
classical receivers  \cite{li2013quantum}. 

 Cleve et al. \cite{cleve1999share} (CGL99)  formulated a fully-quantum threshold scheme and established the first systematic framework for encoding and reconstructing quantum secrets, thereby providing a direct analogue to Shamir's classical scheme in the quantum domain.
 Smith \cite{smith2000quantum} further advanced the theory by extending linear secret-sharing schemes to the quantum setting, enabling constructions for arbitrary access structures.

In the CGL99 QSS protocol %
\cite{cleve1999share}, 
the secret is a single qutrit state,
$%
\ket{\phi} = \alpha \ket{0}  + \beta \ket{1}  + \gamma \ket{2}  %
$, %
encoded into three qutrit shares: %
$\ket{\phi} \mapsto \alpha (\ket{000} + \ket{111} + \ket{222}) + \beta (\ket{012} + \ket{120} + \ket{201}) + \gamma (\ket{021} + \ket{102} + \ket{210})
$ which %
are then distributed among Alice ($\text{Share}_1$), Bob ($\text{Share}_2$), and Charlie ($\text{Share}_3$).

This corresponds to a $(2,3)$ t-QSS scheme, in which two participants or more can perfectly reconstruct the original secret, while a single participant gains no information from their share. 
For instance, if Alice and Bob collaborate, they can apply CNOT operations to recover the secret:
\begin{enumerate}
    \item $\text{Share}_2 \leftarrow \text{Share}_1 + \text{Share}_2 \pmod{3}$,
    \item $\text{Share}_1 \leftarrow \text{Share}_2 + \text{Share}_1 \pmod{3}$.
\end{enumerate}
Upon these operations, the global state becomes
\begin{equation}
\ket{\phi} \otimes (\ket{00} + \ket{12} + \ket{21}).
\end{equation}
Hence, the secret $\ket{\phi}$ is recovered, %
while Charlie gains no information. %

Zou et al.~\cite{zou2015information} studied classical secret sharing over a noisy
broadcast channel, where the dealer encodes the secret and qualified subsets of
participants recover it from their received outputs. Their setting is modeled as an equivalent compound wiretap channel
\cite{liang2009compound,wyner1975wire}, where the dealer communicates a secret to
legitimate receivers while preventing eavesdroppers from gaining any knowledge of it.

 Here, we analyze a QSS model, where %
the dealer distributes quantum information to the participants via a quantum broadcast channel.
We then consider the transformation of this broadcast channel into an equivalent compound quantum channel defined by the collection of channels from the dealer to each qualified set of quantum systems. We observe that due to the no-cloning theorem, the recovery of the quantum secret by the qualified set implies secrecy with respect to the other participants. That is, secrecy is inherent. This abstraction enables us to leverage capacity results for quantum subspace transmission via compound quantum channels, given channel side information at the receiver (CSIR) %
\cite{bjelakovic2008quantum}. 
That is, our model assumes an informed decoder, whose decoding map may depend on the selected 
channel.
This aligns naturally with the structure of
QSS, since the decoding operation depends on the particular qualified subset of quantum receivers. %

The compound quantum channel representation is illustrated in  Fig.~\ref{fig:compound-channel}, for t-QSS, where $t=2$ and $K=3$. The collection of receivers is divided into legitimate receivers and eavesdroppers, corresponding to qualified and non-qualified sets of participants, respectively.
The family of quantum channels is thus associated with the access structure $\mathscr{A}$, where each member is a qualified set of participants. 

Our main capacity result is that the QSS capacity is characterized by the regularization of the following formula: 
\begin{align}
I_c(\mathscr{A})
\;\equiv\;
\max_{\ket{\phi}_{A'A}}
\;\min_{\ell \in \mathscr{A}}
\;
I(A' \rangle B_\ell)%
,
\end{align}%
where $I(A' \rangle B_\ell)$ denotes the coherent information between systems $A'$ and $B_\ell$.
Intuitively, $A'$ represents the dealer's entanglement share, while $B_\ell$  is the joint system of the $\ell$th qualified set.
Furthermore, we derive a closed-form formula for the capacity for t-QSS with dephasing noise.

While we follow a similar approach as in the classical framework of Zou \emph{et al.}~\cite{zou2015information}, its extension to the quantum domain introduces fundamental challenges. In contrast to the classical wiretap setting, secrecy here follows directly from the no-cloning theorem, and coherent information replaces mutual information as the relevant rate measure. The proof technique further departs from the classical setting by employing the teleportation protocol to establish quantum communication through %
entanglement generation. %

\begin{figure}[!tbp]
    \centering
    \includegraphics[width=0.9\linewidth]{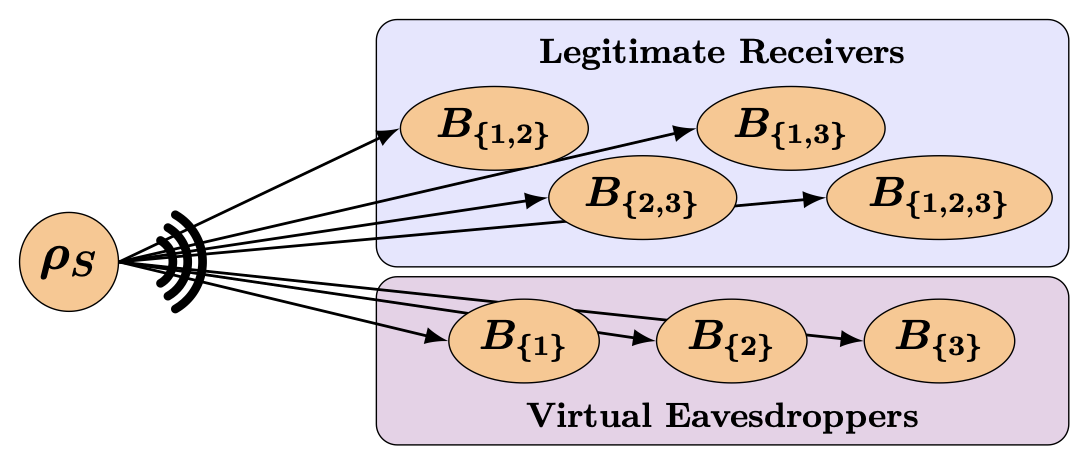}
    \caption{Equivalent compound quantum channel for a t-QSS $(2,3)$ scheme.}
    \label{fig:compound-channel}
\end{figure}

The paper is organized as follows. %
Section~\ref{Section:QSS} presents the formal definition of QSS and an explicit scheme by Smith \cite{smith2000quantum}. In Section~\ref{Section:CQC}, we introduce the compound channel model. In Section~\ref{Section:Main_results}, we state our QSS capacity result.
The analysis is given in Section~\ref{Section:Q_equals_E}, and
a summary in Section~\ref{Section:Conclusion}.

\paragraph*{Notation}
We use the following notation conventions. 
A quantum state $\rho$ is a density operator over a Hilbert space $\mathcal{H}$. We assume $\mathcal{H}$ is finite-dimensional.
Let $\mathcal{S}(\mathcal{H})$ denote the set of all such density operators. %
The  fidelity between two quantum states \( \rho \) and \( \sigma \) in \(\mathcal{S}(\mathcal{H}) \),  is defined by
$%
F(\rho, \sigma) := \norm{ 
\sqrt{\rho}\sqrt{\sigma}
}_1^2
$. 

 The von Neumann entropy is defined as
$%
H(\rho) := -\operatorname{tr}(\rho \log \rho)
$. %
When
associated with system $A$, we also denote %
$H(A)_\rho\equiv H(\rho_A)$.
For a bipartite state $\rho_{AB}$,
the coherent information  %
is then
$
I(A\rangle B)_\rho := H(B)_\rho-H(AB)_\rho %
$. 
\section{Quantum Secret Sharing}
\label{Section:QSS}

We begin with a general description of QSS.
A QSS scheme is a protocol where a sender has an unknown quantum state $\ket{\phi}$ to distribute among a set of participants, for example Alice, Bob, and Charlie. Each participant receives a share of the  output state, $\rho_{ABC}$.   
If a \emph{qualified set} of participants combines their shares together, they can recover the secret $\ket\phi$.  
For instance, in a $(2,3)$ t-QSS scheme, any subset of (at least) two participants can recover the secret. Mathematically, the scheme assigns a decoding map to each authorized collaboration, as shown in the table below:

\begin{table}[ht]
\caption{\small{$(2,3)$ t-QSS}}
\label{Table:2_3_t_QSS}
\centering
\resizebox{\columnwidth}{!}{
\begin{tabular}{|c|c|c|c|}
\hline
\textbf{Channel output} & \textbf{Collaboration} & \textbf{Decoding function} & \textbf{Decoding output} \\
\hline
& Alice and Bob & $\mathcal{D}_{AB} \otimes \identity(\rho_{ABC})$ & $\ketbra{\phi}_{AB} \otimes \rho_{C}$ \\
$\rho_{ABC}$  & Bob and Charlie & $\mathcal{D}_{BC} \otimes \identity(\rho_{ABC})$ & $\rho_{A}\otimes \ketbra{\phi}_{BC}$ \\
 & Alice and Charlie & $\mathcal{D}_{AC} \otimes \identity(\rho_{ABC})$ & $\ketbra{\phi}_{AC} \otimes \rho_{B}$ \\
\hline
\end{tabular}
}%
\end{table}
Ideally, the qualified participants should recover the secret state $\ket\phi$, while the remaining shares become decoupled. This decoupling ensures that the unqualified participants cannot extract information on the secret.

\subsection[Smith's Construction]{Smith's Construction}
We briefly describe the framework by Smith \cite{smith2000quantum}.
\subsubsection{Mathematical framework }
\label{Substion:Smith_Definitions}
Let $[K]=\{1,\ldots,K\}$ denote the set of participants. 

\begin{definition}%
    A \emph{qualified set} $T \subseteq [K]$ is a set of participants that may recover the secret when collaborating.
\end{definition}

\begin{definition}%
An adversary structure %
is a family 
$\mathscr{Z} \subseteq 2^{[K]}$
that is %
\begin{enumerate}[(i)]
    \item Downward-closed:
     $Z \in \mathscr{Z}$ and $Z' \subseteq Z$ implies $Z' \in \mathscr{Z}$.
    
    \item Self-dual:
    $%
        Z \in \mathscr{Z}
        \quad\Longleftrightarrow\quad
        Z^c \notin \mathscr{Z}
    $, for every subset $Z \subseteq [K]$, %
    where $Z^c := [K] \setminus Z$ is its complement. %
\end{enumerate}
The members of $\mathscr{Z}$ are called non-qualified sets.
\end{definition}

\begin{definition}[Access structure]
Given an adversary structure $\mathscr{Z}$, the access structure is defined as
$%
    \mathscr{A}
    := 2^{[K]} \setminus \mathscr{Z} %
$. %
The members of $\mathscr{A}$ are the qualified sets.
 \end{definition}

\begin{remark}[Complementary set]
\label{Remark:Complement_Non_Qualified}
The definition %
implies that the complement of every non-qualified set is qualified, i.e., if $T\in\mathscr{Z}$, then $ T^c \in \mathscr{A}$.
\end{remark}

The scheme consists of the following: %
\begin{enumerate}[(i)]
    \item A secret density operator $\rho_S\in \mathcal{S}(\mathcal{H}_S)$.

    \item An adversary and access structures, $\mathscr{Z}$ and $\mathscr{A}$. %

    \item The decoding operation for every qualified set $T \in \mathscr{A}$. %

\end{enumerate}

\textit{\underline{Algorithm’s construction :}}
The dealer holds a secret density operator $\rho_S$ on the Hilbert space
$\mathcal{H}_S = \mathrm{span}\{\ket{s} : s \in \mathbb{F}_q\}$, where $\mathbb{F}_q$ is a finite field.

\noindent\textbf{Encoding:} 
Denote the overall output space of participants by %
\begin{align}
  \mathcal{H}_{B_{[K]}}
    := \bigotimes_{k\in [K]} \mathcal{H}_{B_k},  
\end{align}
where each subsystem ${B_k}$ is assigned to the $k$th player.

It suffices to define 
the encoding on the basis vectors $\ket{s}$. %
The encoding isometry $F :\; \mathcal{H}_S \to \mathcal{H}_{{B_{[K]}}}$ is defined by $\ket{s} \longmapsto 
               \frac
               {\ket{f(s)}}
               {\norm{f(s)}}$, where \cite{smith2000quantum}
\begin{align}
               & %
               \ket{f(s)}
               :=
               \sum_{\mathbf{a} \in \mathbb{F}_q^{\,e-1}}
               \left|\, M 
                    \begin{pmatrix}
                        s \\[2pt]
                        \mathbf{a}
                    \end{pmatrix}
               \right\rangle 
\end{align}
over  $\mathbb{F}_q$,
where  $M$ is a $K \times e$ matrix, %
$    (\begin{matrix} s & \mathbf{a} \end{matrix})^T
$
is %
formed by concatenation. %
The encoded state is %
$
    \rho_{f(s)}
    :=
    F\rho_S F^{\dagger}$.

\noindent\textbf{Shares:} Player $k$ receives subsystem $B_{k}$. %

\medskip

\noindent\textbf{Decoding:}
Consider a given qualified set, $T \in \mathscr{A}$.  
The global Hilbert space decomposes as
\begin{align}
    \mathcal{H}_{B_{[K]}}
    \;=\;
    \mathcal{H}_{B_{ T}} \otimes \mathcal{H}_{B_{ T^c}}.
\end{align}
where $B_{ T }=\left( B_{k} \right)_{k\in T}$ and $B_{ T^c}=\left( B_{k} \right)_{k\in T^c}$.
Smith~\cite{smith2000quantum} showed that there exists 
a decoding channel 
$%
    \mathcal{D}^{{(T)}} :
    \mathcal{S}(\mathcal{H}_{B_{ T}})
    \to
    \mathcal{S}(\mathcal{H}_S \otimes \mathcal{H}_0)    
$ %
such that 
\begin{align}
    \rho_{S} \otimes \Gamma
    :=
    \bigl(\mathcal{D}^{(T)} \otimes \identity_{B_{ T^c}}\bigr)
    \rho_{f(s)}
\end{align}
for some uncorrelated 
$\Gamma \in \mathcal{S}(\mathcal{H}_0 \otimes \mathcal{H}_{B_{T^c}})$ where $\mathcal{H}_0$ is an auxiliary environment space. %

Thus, any qualified set %
of participants can perfectly reconstruct the secret,
while the others obtain no information.

\begin{remark}[Classical secret sharing]    In the classical construction, the dealer samples the vector 
\(\mathbf{a} \in \mathbb{F}_q^{\,e-1}\) uniformly at random before applying the
linear map \((s,\mathbf{a}) \mapsto M\binom{s}{\mathbf{a}}\). In particular, $M$ is a
 $K\times(t+1)$ Vandermonde matrix for a $(t,K)$ secret sharing scheme \cite{shamir1979share}.
The quantum encoding can be viewed as a coherent version.
\end{remark}

\subsection{QSS Broadcast Model}
\label{Section:QSS_Broadcast}
Suppose
that the dealer transmits the  participants' shares via a quantum broadcast channel.  
If the dealer transmits $A_1\cdots A_n$, %
then the $k$th participant  receives an output sequence $B_{k}^n={B_{k,1}\cdots B_{k,n}}$, for $k\in [K]$. %
Formally, a memoryless broadcast channel is represented by the CPTP map $ \mathcal{N}_{A\to B_1\cdots B_K}^{\otimes n}$, where
$%
    \mathcal{N}_{A\to B_1\cdots B_K} : \mathcal{L}(\mathcal{H}_A) \;\longrightarrow\; 
\mathcal{L}(\mathcal{H}_{B_{[K]}})
$, %
which maps
$\rho_{A}$ to $ \rho_{B_1\cdots B_K}$.
The dealer uses this quantum broadcast channel to distribute a secret 
$\rho_S \in \mathcal{S}( \mathcal{H}_S)$. %

Consider an access structure $\mathscr{A}$ of qualified sets.
In a QSS scheme, access to the secret is guaranteed only for authorized groups. Whenever the players forming a qualified set 
$T \in \mathscr{A}$
 jointly process their respective quantum systems, the original secret state can be faithfully retrieved.
Conversely, %
any set $T \notin \mathscr{A}$ %
cannot recover any information about the secret.

\begin{remark}[No Cloning]
Consider a $(K-1, K)$ t-QSS scheme, where any subset of participants of size $K-1$ or more forms a qualified set. 
Reliability implies that the secret can be recovered from %
$(B_2^n, \ldots, B_K^n)$. 
The remaining participant receives $B_1^n$, and by the no-cloning theorem, they cannot possess any information about the secret. %
This ensures secrecy against all individual eavesdroppers. 
More generally, because the complement of every non-qualified set is
necessarily qualified in our model (see Remark~\ref{Remark:Complement_Non_Qualified}), the same %
argument ensures
secrecy against all unauthorized subsets.

\end{remark}

\begin{definition}[Code for a QSS Broadcast Channel]
\label{Definition:Quantum_Code_QSS}
A \((2^{nR},n)\) code for QSS over \(\mathcal{N}_{A\to B_1\cdots B_K}\) consists of:
\begin{itemize}
\item a secret space $\mathcal{H}_S$ of dimension $2^{nR}$ (for integer $2^{nR}$),

\item an encoding channel $\mathcal{F}_{S\to A^n}:\mathcal{L}(\mathcal{H}_S)\to \mathcal{L}(\mathcal{H}_A^{\otimes n})$, and

\item  a family of decoding channels $\{\mathcal{D}^{{(T)}}_{B_T^n\to \widehat{S}}\}_{T\in\mathscr{A}}$, where \(\mathcal{D}^{{(T)}}_{B_T^n\to \widehat{S}} : \mathcal{L}( \mathcal{H}_{B_T}^{\otimes n}) \rightarrow \mathcal{L}(\mathcal{H}_S)\) produces an estimate of  
 the secret from the  collection $B_T$,  for all %
 $T\in\mathscr{A}$.
\end{itemize}

\end{definition}

The scheme works as follows. 
The dealer holds a ``secret system" $S$ in a state $\rho_S \in \mathcal{S}(\mathcal{H}_S)$,
where $S$ consists of $nR$ qubits. %
Let $\ket{\phi}_{S'S} \in \mathcal{H}_{S'} \otimes \mathcal{H}_S$ be a purification of~$\rho_S$.
Now, the encoder applies the encoding channel: 
\begin{align}
\rho_{S' A^n}= (\identity_{S'}\otimes \mathcal{F}_{S\to A^n})(\ketbra{\phi}_{S'S}).
\end{align}
and sends $A^n$
through $n$ uses of a quantum broadcast channel 
$\mathcal{N}_{A\to B_1\cdots B_K}$, 
which yields
\begin{align}
\rho_{S' B_{[K]}^n %
}=\left(\identity_{S'}\otimes (\mathcal{N}_{A\to B_1 \cdots B_K})^{\otimes n} \right) (\rho_{S' A^n }).
\end{align}

Consider a qualified set of participants $T \subseteq [K]$. The overall output 
$B_{[K]}^n$ %
can then be decomposed to
$(B_T^n,B_{T^c}^n)$. 
The qualified set of participants receives $\rho_{B_T^n}$ and applies the corresponding decoding channel, $\mathcal{D}^{{(T)}}$:
\begin{align}
\rho_{S'\widehat{S} B_{T^c}^n}=\left(\identity_{S'}\otimes \mathcal{D}^{{(T)}}_{B_T^n\to \widehat{S}}\otimes\identity_{T^c} \right)(\rho_{S' B_T^n B_{T^c}^n }).
\end{align}

\begin{definition}[Achievable QSS Rate]
A nonnegative number \(R\) is called an achievable QSS rate for %
an access structure \(\mathscr{A}\) if,
for every \(\varepsilon,\delta>0\) and
\(n \ge n_0(\varepsilon,\delta,R)\),
there exists a \((2^{n(R-\delta)},n)\) code such that
\begin{align}
F\!\left(\ketbra{\phi}_{S'S},\, \rho_{S'\widehat{S}}\right)
\ge 1-\varepsilon
\end{align}
where $\rho_{S'\widehat{S} }$ is the reduced state of $\rho_{S'\widehat{S} B_{T^c}^n}$.
\end{definition}
The definition above also guarantees secrecy, as explained in Subsection~\ref{Subsection:Secrecy} below.

\begin{definition}[QSS Capacity]
The QSS capacity is defined as the supremum of achievable rates $R$ for  %
an access structure \(\mathscr{A}\).
We denote the QSS capacity by $C_{\text{QSS}}(\mathscr{A})$.
\end{definition}

\subsection{Secrecy}
\label{Subsection:Secrecy}
To see that the QSS protocol above ensures secrecy, observe that by 
Uhlmann's theorem, for any purification
$\ket{\rho_{S'\widehat{S} B_{T^c} G}}$ of the output state 
$\rho_{S'\widehat{S} B_{T^c}}$, we have 

\begin{align}
F\!\left(\ketbra{\phi}_{S'S}\otimes\ketbra{\theta}_{B_{T^c} G} ,\, \ketbra{\rho_{S'\widehat{S} B_{T^c} G}}\right) \label{eq:mon_equation} \\
\ge 1-\varepsilon \nonumber
\end{align}
for some $\ketbra{\theta}_{B_{T^c} G}.$ %
By the Fuchs–van de Graaf inequalities \cite[Cor.~9.3.2]{wilde2011classical},
\begin{align}
\norm{\ketbra{\phi}_{S'S}\otimes\ketbra{\theta}_{B_{T^c} G}
-
\ketbra{\rho_{S'\widehat{S} B_{T^c} G}}
}_1 \label{Equation:Error_distance}
\\\leq 2\sqrt{\varepsilon} \nonumber
\end{align}

By the trace-distance monotonicity \cite[Cor.~9.1.2]{wilde2011classical},
\begin{align}
\norm{
\ketbra{\theta}_{B_{T^c} G}
-
\rho_{B_{T^c} G}
}_1\leq 2\sqrt{\varepsilon}
\label{Equation:theta_rho_distance}
\end{align}

Next, by the triangle inequality,
\begin{align}
&\norm{\ketbra{\phi}_{S'S}\otimes \rho_{B_{T^c} G} - \ketbra{\rho_{S'\widehat{S} B_{T^c} G}}}_1 \label{eq:ma_preuve} \\
&\leq \Big\lVert\ketbra{\phi}_{S'S}\otimes \rho_{B_{T^c} G} -\ketbra{\phi}_{S'S}\otimes\ketbra{\theta}_{B_{T^c} G} \Big\rVert_1 \nonumber \\
&\phantom{=} +\Big\lVert \ketbra{\phi}_{S'S}\otimes\ketbra{\theta}_{B_{T^c} G} - \ketbra{\rho_{S'\widehat{S} B_{T^c} G}} \Big\rVert_1 \nonumber \\
&= \Big\lVert \rho_{B_{T^c} G}- \ketbra{\theta}_{B_{T^c} G} \Big\rVert_1 \nonumber +\Big\lVert \ketbra{\phi}_{S'S}\otimes\ketbra{\theta}_{B_{T^c} G} 
\\&\phantom{=} - \ketbra{\rho_{S'\widehat{S} B_{T^c} G}} \Big\rVert_1 \nonumber \\
&\leq 4\sqrt{\varepsilon} \nonumber
\end{align}
where the last inequality holds by \eqref{Equation:Error_distance} and \eqref{Equation:theta_rho_distance}.
Applying trace-distance monotonicity once more, we obtain a bound for the indistinguishability with respect to the product state  $\ketbra{\phi}_{S'S}\otimes \rho_{B_{T^c} }$:
\begin{align}
\frac{1}{2}\norm{\ketbra{\phi}_{S'S}\otimes \rho_{B_{T^c} }
-
\rho_{S'\widehat{S} B_{T^c} }
}_1
&\leq
2\sqrt{\varepsilon}
\end{align}
Then, by the Fuchs-van de Graaf inequalities, 
\begin{align}
F\!\left(\ketbra{\phi}_{S'S} \otimes \rho_{B_{T^c} } ,\, \rho_{S'\widehat{S} B_{T^c} }\right)
&\ge 
1-4\sqrt{\varepsilon}
\end{align}
which guarantees secrecy {\cite[Cor.~9.3.1]{wilde2011classical}}.

\begin{remark}[Entanglement with the dealer]
One may view the coding scheme above as a procedure for generating entanglement between the dealer and a qualified set of participants. Specifically, suppose that the dealer first prepares a bipartite entangled state $\ket{\phi}_{S'S}$ locally, where $S$ contains the secret and $S'$ is a resource that the dealer keeps. The dealer then encodes $S$ into $A^n$ and transmits it through the broadcast channel $\mathcal{N}_{A \to B_1 \cdots B_K}^{\otimes n}$. When a qualified set of participants, $T$, wishes to recover the secret, its members cooperate and apply the decoding map $\mathcal{D}^{(T)}_{B_T^n \to \tilde{S}}$. At the end of the protocol, the dealer and the qualified set approximately share the bipartite state $\phi_{S'S}$.
\end{remark}

\begin{remark}[Entanglement monogamy]
According to the monogamy of entanglement, if two systems are maximally entangled, then they cannot be entangled with any third system \cite{terhal2004entanglement} \cite[Sec. 4.5]{vidick2023introduction}. Equivalently, they must be in a product state with the rest of the universe. In our setting, if the dealer and a qualified set share an approximately pure entangled state, then they cannot simultaneously share entanglement with any other parties. Intuitively, this implies that any unqualified set of participants is effectively decoupled from both the dealer's reference system and the secret.

\end{remark}

\section{Coding Definitions}
\label{Section:CQC}

Our model relies on the notion of a \textit{compound quantum channel}, which describes communication over an unknown channel, taken from a family of possible channels. We extend and adapt the results and definitions from Bjelaković \textit{et al.}~\cite{bjelakovic2008quantum} to the context of a QSS  scheme. 

In this section, we provide the definitions for the quantum capacity  and the entanglement-generation capacity with an informed decoder (see Remark~\ref{Remark:CSIR} below). In the following section, we prove that these two quantities are equal, which leads us to our main theorem.

A compound quantum channel with informed decoder is defined by a family of memoryless quantum channels,
\begin{align}
    \left\{ \mathcal{N}^{(\ell)}_{A \to B_\ell} : \mathcal{L}(\mathcal{H}_A) \rightarrow \mathcal{L}(\mathcal{H}_{B_\ell}) \right\}_{\ell \in \mathcal{J}}
\end{align}
where we allow each $B_\ell$ to be of a different dimension.
The channels are indexed by $\ell$ from a finite index set $\mathcal{J}$.

The sender does not know which channel from this set is actually used for transmission,
whereas the decoder is informed. For our purposes, it suffices to consider a finite family, i.e.,  $\abs{\mathcal{J}}<\infty$. One may then assume without loss of generality that the output dimension is identical for all $\ell$, 
and set it as the maximal output dimension.

\subsection{Quantum communications}

\begin{definition}[Compound Quantum Code]
\label{Definition:Quantum_Code}
    A \((2^{nR},n)\) code for a compound channel with an informed decoder consists of:
\begin{itemize}
\item a quantum message space $\mathcal{H}_S$ of dimension $2^{nR}$, %

\item an encoding channel $\mathcal{F}_{S\to A^n}:\mathcal{L}(\mathcal{H}_S)\to \mathcal{L}(
\mathcal{H}_A^{\otimes n} )$
, and
\item  a family of decoding channels $\{\mathcal{D}^{{(\ell)}}_{B^n\to \widehat{S}}\}_{\ell \in \mathcal{J}}$, where \(\mathcal{D}^{{(\ell)}}_{B^n\to \widehat{S}}: \mathcal{L}(\mathcal{H}_B^{\otimes n}) \rightarrow \mathcal{L}(\mathcal{H}_S)\) represents the recovery operation for the corresponding  channel \(\mathcal{N}^{(\ell)}_{A \to B_\ell}\).

\end{itemize}
\end{definition}
The scheme works as follows. 
Alice has a "message" system $S$, in a state 
$\rho_S$. Let $\ket{\phi}_{S'S}$ be a purification of this state. 
She applies the encoding channel, producing the input state
\begin{align}
\rho_{S' A^n}= (\identity_{S'}\otimes \mathcal{F}_{S\to A^n})(\ketbra{\phi}_{S'S})
\end{align}
and sends $A^n$ through $n$ channel uses. Suppose that $\mathcal{N}^{(\ell)}_{A \to B_\ell}$ is the \emph{actual} channel. 
Then, the output state is
\begin{align}
\rho_{S' B_\ell^n }^{(\ell)}=\left(\identity_{S'}\otimes (\mathcal{N}^{(\ell)}_{A\to B_\ell})^{\otimes n}\right)(\rho_{S'A^n}). 
\end{align}
The decoder receives $B_\ell^n$. Recall that we assume an informed decoder, i.e., he knows the channel. 
Hence, he applies the corresponding decoding channel 
\begin{align}
\rho_{S' \widehat{S}}^{(\ell)}=(\identity_{S'}\otimes \mathcal{D}^{{(\ell)}}_{B^n\to \widehat{S}})(\rho_{S' B_\ell^n }^{(\ell)}).
\end{align}

\begin{definition}[Achievable Quantum Rate]
A nonnegative number \(R\) is called an achievable quantum rate for the compound quantum channel with an informed decoder if
for every $\varepsilon,\delta>0$ and 
$n\geq n_0(\varepsilon,\delta,R)$,
there exists a  \((2^{n(R-\delta)},n)\) code such that:
\begin{align}
\min_{\ell\in\mathcal{J}} F\left(\ketbra{\phi}_{S'S} \,,\; \rho_{S' \widehat{S}}^{(\ell)}\right)\geq 1-\varepsilon
\end{align}
for all $\ket{\phi}_{S'S}\in \mathcal{H}_{S'}\otimes \mathcal{H}_S$. %
\end{definition}

\begin{definition}[Capacity with Informed Decoder]
The quantum capacity \(Q(\mathcal{J})\)  is  the supremum of achievable rates $R$ for the  compound quantum channel with an informed decoder.
\end{definition}

\begin{remark}[Side information]
\label{Remark:CSIR}
In models with channel uncertainty, the availability of side information plays a crucial role \cite{Caire1999Capacity}. Various works assume that such information is available at the encoder, the decoder, or both \cite{cai2004quantum}.
In this work, we assume that the decoder is informed, i.e., only the decoder knows which channel from the family $\{\mathcal{N}^{(\ell)}\}_{\ell \in \mathcal{J}}$ is realized. This assumption aligns naturally with the QSS setting, where the identity of the qualified set is known to the participating parties during the recovery phase.
\end{remark}

\subsection{Entanglement Generation}
One may also consider %
entanglement generation.  The code is defined in the same manner as in Definition~\ref{Definition:Quantum_Code}. However, we focus on a maximally entangled state $\ket{\Phi}\in\mathcal{H}_S^{\otimes 2}$,
\begin{align}
\ket{\Phi}\equiv \frac{1}{\sqrt{d_S}}
\sum_{i=0}^{d_S-1}\ket{i}_{S'}\otimes \ket{i}_S
\end{align}
where $d_S\equiv 2^{nR}$ is the entanglement %
dimension. 
Here, $S'$ is the entanglement resource that Alice keeps, and $S$ is the resource that she distributes to Bob.

\begin{definition}[Achievable Entanglement-Generation Rate]
An achievable entanglement-generation rate \(R_{\text{EG}}>0\)
for the compound quantum channel with an informed decoder 
is such that
for every $\varepsilon,\delta>0$ and 
$n\geq n_0(\varepsilon,\delta,R)$,
there exists a  \((2^{n(R_{\text{EG}}-\delta)},n)\) code such that:
\begin{align}
\min_{\ell} F(\ketbra{\Phi}_{S'S},\rho_{S' \widehat{S}}^{(\ell)})\geq 1-\varepsilon.
\end{align}
\end{definition}

The entanglement-generation capacity $E(\mathcal{J})$ is defined accordingly.

\begin{remark}
\label{Remark:Entanglement_Generation}
A quantum communication code can also perform entanglement generation. Hence, if a rate $R$ is achievable for quantum communication, then %
$R_{\text{EG}}=R$ is achievable as well.
It thus follows that $E(\mathcal{J}) \geq Q(\mathcal{J}).$
Later, we will show that the capacities are in fact identical. 
\end{remark}

The following Theorem is  due to Bjelakovi\'c et al. \cite{bjelakovic2008quantum}.
Define
\begin{align}
I_c(\mathcal{J})
\;\equiv\;
\max_{\ket{\phi}_{A'A}}
\;\min_{\ell \in \{1,\cdots,|\mathcal{J}|\}}
\;
I(A' \rangle B_\ell)_{\sigma^{(\ell)}},
\end{align}
where $ \sigma^{(\ell)}_{A' B_\ell} = (\identity_{A'} \otimes \mathcal{N}^{(\ell)}_{A \to B_\ell})(\phi_{A'A}).$
\begin{theorem}[See {\cite[Lemm. 4.3 and Th. 4.3]{bjelakovic2008quantum}}]
\label{Theorem:Bjelakovic}
Let \(   \{ \mathcal{N}^{(\ell)}\}_{\ell\in \mathcal{J}}\) be a compound channel. 
Then,
\begin{align}
E(\mathcal{J}) &\;\ge\; I_c(\mathcal{J}).
\label{eq:entanglement-ineg}
\end{align}
Furthermore, 
\begin{align}
E(\mathcal{J}) &\;=\; \lim_{n\to\infty} \frac{1}{n} I_c(\mathcal{J}^{\otimes n}).
\label{eq:entanglement-limit}
\end{align}
\end{theorem}

\section{Main Results}
We now present our main results on the QSS capacity. 
We begin with the quantum capacity of the compound channel with an informed decoder. 
\label{Section:Main_results}
\begin{theorem}
\label{Theorem:Q_equals_E}
Let \(   \{ \mathcal{N}^{(\ell)}\}_{\ell\in \mathcal{J}}\) be a compound channel. 
Then,
\begin{align}
Q(\mathcal{J})=E(\mathcal{J})
\end{align}
\end{theorem}
See proof %
in Section~\ref{Section:Q_equals_E}.
Our main result is given below, as a consequence of 
Theorem~\ref{Theorem:Bjelakovic} (from \cite{bjelakovic2008quantum})
and
Theorem~\ref{Theorem:Q_equals_E}. %
Next, we follow a similar approach to that in \cite{zou2015information, liang2009compound}.

We construct an \emph{equivalent compound channel} associated with the access structure.
\begin{definition}[QSS Compoud Channel]
\label{Definition:Equivalence}
Consider a QSS scheme over a quantum broadcast channel $\mathcal{N}_{A\to B_1,\ldots,B_K}$ (see Subsection~\ref{Section:QSS_Broadcast}). 
Given an access structure $\mathscr{A}$, we associate to each qualified set 
$T \in \mathscr{A}$ a virtual legitimate receiver 
$B_T := (B_k)_{k \in T}$.  
This defines a compound quantum channel 
\begin{align}
\left\{\mathcal{N}_{A\to \overline{B}_\ell}^{(\ell)}\right\}_{\ell \in \mathcal{J}}
\end{align}
with one virtual
receiver $\ell$ per qualified set $T$, where $\overline{B}_\ell \equiv B_T$. %
Hence, $\mathcal{J}\cong\mathscr{A}$.
\end{definition}

\begin{remark}[Compound channel for $(2,3)$ t-QSS]
\label{Remark:2_3_QSS}
Consider for example the $(2,3)$ t-QSS scheme, where any subset of at least two participants can recover the secret (see Table~\ref{Table:2_3_t_QSS}).
Here, the dealer sends $A=(A_1,A_2 ,A_3)$ to the players through a quantum communication channel, 
and the players, Alice, Bob, and Charlie receive  $B_1$, $B_2$, and $B_3$, respectively.
The access structure is thus associated with a compound quantum channel $\{ \mathcal{N}^{(\ell)}_{A\to \overline{B}_\ell} \}_{\ell\in\mathcal{J}}$, with an index set 
\begin{align*}
\mathcal{J}\cong\mathscr{A}=\big\{\{\text{Alice, Bob}\},\{\text{Alice, Bob}\},\{\text{Alice, Charlie}\}, \\\{\text{Bob, Charlie}\},\{\text{Alice, Bob, Charlie}\}\big\}.
\end{align*}
 We may define $\mathcal{J}=\{1,2,3,4\}$, and consider the associated family of four channels:
\begin{align*}
\{  \mathcal{N}^{(\ell=1)}_{A\to B_1 B_2} \,,\; \mathcal{N}^{(\ell=2)}_{A\to B_1 B_3} 
\,,\; \mathcal{N}^{(\ell=3)}_{A\to B_2 B_3} \,,\; \mathcal{N}^{(\ell=4)}_{A\to B_1 B_2 B_3} \}.
\end{align*}
Here, $\overline{B}_1=(B_1,B_2)$, $\overline{B}_2=(B_1,B_3)$, $\overline{B}_3=(B_2,B_3)$, and $\overline{B}_4=(B_1,B_2,B_3)$. 
\end{remark}

\begin{remark}[Informed decoder]
Recall that we consider a compound channel model with an informed decoder, i.e., the decoder knows the index $\ell \in \mathcal{J}$ of the channel that is realized (see Remark~\ref{Remark:CSIR}). This assumption is well aligned with the QSS setting.
Indeed, when the participants of a qualified set $T$ choose to cooperate, they know that the effective channel acting on their joint system is $\mathcal{N}_{A \to B_T}^{(T)}$. Consequently, they can apply a decoding map tailored to this specific channel. In this sense, the decoder is effectively informed, since the recovery operation in QSS is explicitly conditioned on the identity of the qualified set, meaning that the decoding procedure depends on $T$.
\end{remark}

\begin{theorem}
\label{Theorem:CQSS_equals_Q}
Let $\mathscr{A}$ be an access structure, and $\mathcal{J}$ be associated with a compound quantum channel as specified in Definition~\ref{Definition:Equivalence}. Then, 
\begin{align}
C_{\text{QSS}}(\mathscr{A})=Q(\mathcal{J}).
\end{align}
\end{theorem}

\begin{proof}
A code for the compound quantum channel allows every legitimate receiver to
recover the message while the environment gains no information. When applied
to the induced compound channel associated with the access structure, each
qualified set $T$ is represented by a virtual receiver $B_T$ that can decode the
secret, whereas the complementary set $T^c$ obtains no information (see Subsction~\ref{Subsection:Secrecy}). This matches exactly the requirements of a QSS scheme, hence
$C_{\text{QSS}}(\mathscr{A})=Q(\mathcal{J})$.
\end{proof}

Based on Theorem~\ref{Theorem:Bjelakovic} 
due to Bjelakovi\'c et al. \cite{bjelakovic2008quantum}, 
Theorem~\ref{Theorem:Q_equals_E}, and Theorem~\ref{Theorem:CQSS_equals_Q}, we obtain the following corollary.
\begin{corollary} 
The QSS capacity for an access structure $\mathscr{A}$ satisfies
\begin{align}
C_{\text{QSS}}(\mathscr{A}) &\;\ge\; I_c(\mathscr{A}).
\label{eq:QSS-ineg}
\end{align}
Furthermore, 
\begin{align}
C_{\text{QSS}}(\mathscr{A}) &\;=\; \lim_{n\to\infty} \frac{1}{n} I_c(\mathscr{A}^{\otimes n})
\end{align}
\end{corollary}

If the degradedness relation is a total order over $\mathscr{A}$ and
every channel $\mathcal{N}^{(\ell)}$ is degradable, for $\ell\in\mathscr{A}$, then  
$C_{\text{QSS}}(\mathscr{A}) = I_c(\mathscr{A})$ \cite{smith2008private}. 
An example is given below.

\begin{example}[Dephasing Channel on a (2,3) t-QSS]
Consider a (2,3) t-QSS scheme with three players.
As explained in Remark~\ref{Remark:2_3_QSS}, the associated access structure can be written as
$\mathscr{A}=\{\mathtt{ab},\mathtt{ac},\mathtt{bc},\mathtt{abc}\}$, where $\mathtt{a}$, $\mathtt{b}$, and $\mathtt{c}$ stand for  Alice, Bob, and Charlie.

Let $\mathcal{N}_{A\to B_\ell}^{(\ell)}$ be a dephasing channel,
\begin{align}
\mathcal{N}_{A\to B_\ell}^{(\ell)}(\rho)
= [1-q(\ell)]\,\rho + q(\ell)\, Z\rho Z^\dagger ,
\end{align}
where $Z$ is the Pauli phase-flip operator and $q(\ell)\in[0,1]$ is the dephasing parameter, for 
$\ell\in\mathscr{A}$. %

Therefore, Theorem~\ref{Theorem:CQSS_equals_Q} gives us 
\begin{align}
C_{\text{QSS}}(\mathscr{A})
&= \log 3 - \max_{q \in \{q(\mathtt{ab}), q(\mathtt{ac}), q(\mathtt{bc}), q(\mathtt{abc})\}} H_2(q).
\end{align}
Achievability is obtained by taking the maximally entangled input state, $\ket{\phi_{A'A}}=\frac{1}{\sqrt3}\sum_{i=0}^2 \ket{i}\otimes \ket{i}$. The converse proof follows as in \cite[Ex. 24.7.1]{wilde2013quantum}. 
\end{example}

\section{Proof of Theorem~\ref{Theorem:Q_equals_E}}
\label{Section:Q_equals_E}

We already know that $E(\mathcal{J}) \geq Q(\mathcal{J})$, as pointed out in Remark~\ref{Remark:Entanglement_Generation}. 
It remains to show that $E(\mathcal{J}) \leq Q(\mathcal{J})$.
By the teleportation protocol, generating entanglement at rate $R$ allows quantum communication at the same rate, provided that classical communication is available. Thus, it remains to show that classical assistance is unnecessary for transmitting quantum information over a compound quantum channel with an informed decoder.

Consider a code that transmits quantum information at rate $R$ with classical
communication assistance. We aim to show that the same rate is achievable
without this assistance. Suppose that Alice’s encoder is a quantum instrument of
the following form:
\begin{align}
    \mathcal{F}(\phi) = \sum_m   \mathcal{F}_m(\phi) \otimes \ketbra{m},
\end{align}
with $\phi\equiv \ketbra{\phi}_{S'S}$,
where %
$m$ is the classical message that Alice sends to Bob using the classical assistance.  
The corresponding output state is then
$%
    \sum_m \mathcal{D}_m^{(\ell)}\circ\mathcal{N}^{(\ell)\otimes n}\circ \mathcal{F}_m(\phi)
$, %
where $\mathcal{N}^{(\ell)}$ is the actual channel.

By hypothesis, reliable quantum communication at rate $R$ with classical assistance implies the existence of a sequence of encoders and decoders $\{\mathcal{F}_{m},\, \mathcal{D}_{m}^{(\ell)}\}$ and error parameters $\varepsilon_n$ such that
\begin{align}
1-\varepsilon_n \leq \min_\ell F\!\left(
\sum_m \mathcal{D}_m^{(\ell)}\circ \mathcal{N}^{(\ell)\otimes n}\circ \mathcal{F}_m(\phi) ,
\, \ketbra{\phi}
\right)
\end{align}
where $\varepsilon_n$ tends to zero as $n\to\infty$.
It follows that 
\begin{align}
&1-\varepsilon_n
\nonumber\\
&\leq \frac{1}{\abs{\mathcal{J}}}\sum_\ell F\!\left(
\sum_m \mathcal{D}_m^{(\ell)}\circ \mathcal{N}^{(\ell)\otimes n}\circ \mathcal{F}_m(\phi) ,
\, \ketbra{\phi}
\right)
\nonumber\\
&=  \frac{1}{\abs{\mathcal{J}}}\sum_\ell \langle \phi |\left(
\sum_m \mathcal{D}_m^{(\ell)}\circ \mathcal{N}^{(\ell)\otimes n}\circ \mathcal{F}_m(\phi)\right)
\ket{\phi}
\nonumber\\
&= \sum_m p_M(m)
\left[  \frac{1}{\abs{\mathcal{J}}}\sum_\ell
\bra{ \phi}\,
\mathcal{D}_m^{(\ell)}\circ \mathcal{N}^{(\ell)\otimes n}\circ \mathcal{F}_m'(\phi)
\ket{\phi} \right]
\end{align}
where $\mathcal{F}'_m$ is a rescaling of %
$\mathcal{F}_m$, and $p_M(m)$ is the problability that Alice has sent the classical message $m$.
Then, there exists $m^*$ such that
\begin{align}
\frac{1}{\abs{\mathcal{J}}}\sum_\ell \abs{ 
1-\bra{\phi}\,
\mathcal{D}_{m^*}^{(\ell)}\circ \mathcal{N}^{(\ell)\otimes n}\circ \mathcal{F}'_{m^*}(\phi)
\ket{\phi}
}
&\le  \varepsilon_n.
\end{align}
It follows that
\begin{align} 
1-F\left(
\mathcal{D}_{m^*}^{(\ell)}\circ 
    \mathcal{N}^{(\ell)\otimes n}\circ 
        \mathcal{F}'_{m^*}(\phi) 
\,,\;
\ketbra{\phi}
\right) 
&\le  \abs{\mathcal{J}} \varepsilon_n
\end{align}
for all $\ell$. That is,
\begin{align*} 
\min_\ell F\left(
\mathcal{D}_{m^*}^{(\ell)}\circ 
    \mathcal{N}^{(\ell)\otimes n}\circ 
        \mathcal{F}'_{m^*}(\phi) 
\,,\;
\ketbra{\phi}
\right)
&\ge  1-\abs{\mathcal{J}} \varepsilon_n.
\end{align*}

Having assumed that $\abs{\mathcal{J}}$ is finite,
this %
tends to $1$ as 
$n\to \infty$.

Consequently, Alice and Bob can  use the encoding-decoding pair $(\mathcal{F}'_{m^*}, \mathcal{D}^{(\ell)}_{m^*})$, %
where  $m^*$ is fixed. %
This shows that the forward classical communication of the classical message $m$ is not needed and  the same performance can be achieved without such classical assistance.
\qed %

\section{Conclusion}
\label{Section:Conclusion}

In this work, we introduced an information-theoretic framework for quantum secret sharing over noisy broadcast channels and defined the corresponding QSS capacity. By translating the QSS access structure to a compound quantum channel with an informed decoder, we showed that the QSS capacity equals the quantum capacity of the induced compound channel, leading to a regularized coherent-information characterization, which fundamentally differs from the classical secret sharing characterization. A key feature of the model is that secrecy follows directly from the no-cloning theorem once reliable reconstruction is ensured.
Future research directions include the design of explicit and efficient %
QSS schemes approaching capacity, extension of our approach to other cryptographic primitives such as bit commitment \cite{hayashi2022commitment}, oblivious transfer \cite{yang2025quantum} and authentication \cite{farre2025entanglement}, %
continuous-variable settings, and quantum sharing of classical secrets \cite{hillery1999quantum}.

\section*{Acknowledgment}
This work was supported by  ISF, %
 Grants n. 939/23 and 2691/23,
 DIP %
 n. 2032991, Ollendorff-Minerva Center %
 n. 86160946, and  %
 HD Quantum Center %
 n. 	2033613.

\bibliographystyle{IEEEtran}
{\balance
\bibliography{references}
}

\end{document}